\documentclass[copyright,creativecommons]{eptcs}
\providecommand{\event}{EXPRESS/SOS2021} 
\usepackage{amsmath}
\usepackage{amsthm}
\usepackage{amssymb}
\usepackage{esvect}
\usepackage{paralist}
\usepackage[capitalise]{cleveref} 
\usepackage[all]{xy}
\usepackage{graphicx}
\usepackage{xcolor}
\title{Minimal Translations from Synchronous Communication to Synchronizing Locks}

\author{
Manfred Schmidt-Schau{\ss}
\institute{Goethe-University, Frankfurt am Main, Germany}
\email{schauss@ki.cs.uni-frankfurt.de}
\and 
David Sabel 
\institute{LMU, Munich, Germany} 
\email{david.sabel@ifi.lmu.de}
}

\def\titlerunning{Minimal Translations from Synchronous Communication to Synchronizing Locks}
\def\authorrunning{M. Schmidt-Schau{\ss}, D. Sabel}

\newcommand{\ignore}[1]{}

\newcommand{\maycon}{{\downarrow}}

\newcommand{\wrt}{{w.r.t.}}

\newcommand{\PAR}{{\ensuremath{\hspace{0.1mm}{\!\scalebox{2}[1]{\tt |}\!}\hspace{0.1mm}}}}

\newcommand{\dstodo}[2][]{{\color{blue} \bfseries DAVID: #2}}

\newcommand{\msstodo}[2][]{{\color{cyan} \bfseries MANFRED: #2}}

\newenvironment{proof*}{{\it Proof.}}{}
\newtheorem{theorem}{Theorem}[section]
\newtheorem{lemma}[theorem]{Lemma} 
\newtheorem{example}[theorem]{Example} 
\newtheorem{definition}[theorem]{Definition} 
 
\newtheorem{proposition}[theorem]{Proposition}
\newtheorem{corollary}[theorem]{Corollary}

\usepackage{esvect}
\usepackage[all]{xy}

\begin{document}
\maketitle
 
\newcommand{\bbbn}{\mathbb{N}}
\newcommand{\CHSIMPLE}{\ensuremath{\text{CHSIMPLE}}}
\newcommand{\LOCKSIMPLE}{\ensuremath{\mathrm{LOCKSIMPLE}}}
\newcommand{\PISIMPLE}{\ensuremath{\mathrm{{SYNCSIMPLE}}}}
\newcommand{\PIS}{\mathit{SYS}}
\newcommand{\TransCheckTool}{\text{Refute-Pi}} 
\newcommand{\RegCheckTool}{\text{Refute-Regex}}
\newcommand{\gentzen}[2]{{\displaystyle \frac{#1}{#2}}}
\newcommand{\gentzent}[3]{{\displaystyle \frac{#1}{#2}}\quad #3}

\newcommand{\success}{\checkmark}
\newcommand{\silent}{0}


\newcommand{\full}{{\ensuremath{\blacksquare}}}

\newcommand{\eempty}{{\ensuremath{\Box}}}

\begin{abstract}
In order to understand the relative expressive power of larger concurrent programming languages, we analyze translations of small process calculi which model the communication and synchronization of concurrent processes. The source language SYNCSIMPLE is a minimalistic model for message passing concurrency while the target language LOCKSIMPLE is a minimalistic model for shared memory concurrency. The former is a calculus with synchronous communication of processes, while the latter has synchronizing mutable locations -- called locks -- that behave similarly to binary semaphores. The criteria for correctness of translations is that they preserve and reflect may-termination and must-termination of the processes. We show that there is no correct compositional translation from SYNCSIMPLE to LOCKSIMPLE that uses one or two locks, independent from the initialisation of the locks. We also show that there is a correct translation that uses three locks. Also variants of the locks are taken into account with different blocking behavior.
\end{abstract}


\sloppy

%

 \section{Introduction}


Different models of concurrency are studied and used in theory and in practice of computer science. One main approach are \emph{message passing models} where the concurrently running threads 
(or processes) communicate by sending and receiving messages.
A prominent example for a message passing model is the $\pi$-calculus \cite{milner1992calculus,sangiorgi-walker:01}.
There exist approaches with asynchronous and with synchronous message passing. 
In asynchronous message passing, a sender sends its message and proceeds \emph{without} waiting that a receiver collects the message (thus the message is kept in some medium until the receiver collects it from that medium). 
In synchronous message passing, the message is exchanged in one step and thus sender and receiver wait until the communication 
has happened. Thus, synchronous message passing can be used for synchronization of processes. 

Another main approach for concurrency are program calculi with \emph{shared memory} where concurrent processes communicate by using shared memory primitives. For instance, $\lambda(\textbf{fut})$ \cite{jlambda-fut} is a program calculus that models the core of  the strict concurrent functional language Alice ML, and it has concurrent threads, handled futures, and memory cells with an atomic exchange-operation. Also other shared memory synchronization primitives like concurrent buffers  and their encodability into  $\lambda(\textbf{fut})$  are analyzed \cite{schwinghammer-sabel-schmidt-schauss-niehren:09:ml}.
Other examples are the calculi  CH \cite{schmidt-schauss-sabel-2020} and CHF \cite{sabel-schmidt-schauss-PPDP:2011,sabel-schmidt-schauss-LICS:12,schauss-sabel-dallmeyer:18}. 
The latter is a program calculus that models the core of Concurrent Haskell \cite{peyton-gordon-finne:96}: it extends the functional programming language Haskell by concurrent threads and so-called MVars, which are synchronizing mutable memory locations. Thus, depending on the model (or the concurrent programming languages) there exist different primitives. 
The simplest approach is some kind of locking primitive to block a process until some event happens.
 To exchange a message, for instance, atomic read-write registers can be used. More sophisticated primitives are for example semaphores, monitors, or Concurrent Haskell's MVars. 
 All these approaches have in common that processes can be blocked until an event occurs, which is performed by another process.

Expressivity of (concurrent) programming languages is an important topic, since
the corresponding results allow us to classify the languages and their programming constructs, and to understand their differences. 
Investigating the expressivity  to clarify the relation between
message passing models and shared memory concurrency can  in principle 
be done by constructing correct translations from one model to the other.
Our research considers the question whether and how synchronous message passing can be implemented by models that support shared memory and some of these synchronization primitives.

In previous work \cite{schmidt-schauss-sabel-2020}, we analyzed translations from the synchronous $\pi$-calculus into a core language of  Concurrent Haskell. In particular, we looked for compositional translations that preserve and reflect 
the convergence behavior of processes (in all program-contexts) {w.r.t.}~may- and should-convergence.
This means, processes can successfully terminate or not, where may-convergence observes whether there is a possible execution path to a successful process and should-convergence means that the ability to become successful holds for all execution paths. 
We found correct translations and proved them to be correct with respect to this correctness notion. 
Looking for small translations has several advantages: The resource usage of the translated programs is lower, they  are easier to understand than larger ones, 
and the corresponding correctness proofs often are easier than for large ones. Hence, we also tried to find smallest translations, but in the end
we could not answer the following question:
what is the minimal number of MVars that are necessary to correctly encode the message passing synchronization using MVars?
This leads us to the general question how synchronous communication can be encoded by synchronizing primitives and what is the minimal number of primitives that is required. 
This question is addressed in this paper.
We choose to work with models that are as simple as possible and also as complex as needed, 
but nevertheless are also relevant for full programming languages
(we discuss the transportion of the results to full languages in
\cref{subsec:correct-translations}).
%
Thus we consider translations from a small message passing source language into a small target language with shared memory concurrency and synchronizing primitives.

For the source language $\PISIMPLE$, we use a minimalistic model for concurrent processes 
that synchronize by communication. The language has constructs for sending (denoted by ``!'') and for receiving (denoted by ``?''). A communication step 
atomically processes one $!$ from one process together with one $?$ from another process. For simplicity, there is no message that is sent and there are 
no channel names ({i.e.}~the language can be seen as a variant of the synchronous $\pi$-calculus (without replication and sums) where only one global channel name exists).

For the target language $\LOCKSIMPLE$ we choose a similar calculus where the communication is removed and replaced by synchronizing shared memory locations. 
These locations are called \emph{locks}. 
 A lock can be empty or full. There are operations to fill an empty lock (put) or to empty a full lock (take). The main variant that we consider is the one 
 where the put-operation blocks on a full lock, but the take-operation is not blocking on an empty lock.
Thus these locks are like binary semaphores where put is the wait-operation and take is the signal-operation (where signal on an unlocked semaphore is allowed but has no effect).
We also consider the language with several locks with different initializations (empty or full).
Based on this setting, the question addressed by the paper is:
\begin{quote}
\emph{What is the minimal number of locks that is  required to correctly translate the source calculus into the target calculus?}
\end{quote}

The notion of correctness of a translation requires comparing the semantics of both calculi. We adopt the approach of observational correctness \cite{schmidt-schauss-niehren-schwinghammer-sabel-ifip-tcs:08,schmidtschauss-sabel-niehren-schwing-tcs:15} and thus we use correctness {w.r.t.}~a~contextual equivalence which considers the may- and the must-convergence in both calculi. May-convergence means that the process can be evaluated to a successful process 
(in both calculi we add a constant to signal success). Due to the nondeterminism, observing may-convergence is too weak since for instance, it equates processes 
that must become successful with processes that either diverge or become successful. Hence we also observe must-convergence, which holds if any evaluation of the process ends 
with a successful process. Considering  must-convergence only is also too weak since it equates processes that always fail with processes that either fail or become successful. 
Thus we use the  combination of both convergencies as program semantics. In turn, a correct translation must preserve and reflect the may- and must-convergence of any program.  

This can also be seen as a minimalistic requirement on a correct translation since for instance, requiring equivalence of strong or weak bisimulation (see {e.g.}~\cite{sangiorgi-walker:01}) would be a much stronger requirement.


{\em Results.}
We show that there does not exist a correct compositional   
translation from $\PISIMPLE$ into $\LOCKSIMPLE$ that uses one (\cref{thm:storage-1lock-impossible}) or two locks (\cref{thm:tau-2-incorrect}), while there is a correct 
compositional translation that uses three locks (\cref{thm:translation-k-3-len-6}).

 The non-existence is proved for any initial state of the lock variables and also for different kinds of blocking behaviour of the lock ({i.e.}~whether {put} or whether~{take} blocks).

{\em Related Work.}
Validity of translations between process calculi is discussed in \cite{Gorla:10,GlabbeekGLM19} 
where five criteria for valid translations resp.~encodings are proposed:
compositionality, name invariance, operational correspondence, divergence reflection, and success sensitiveness. Compositionality and name invariance restrict the syntactic form of the translated processes;  operational correspondence means that the transitive closure of the reduction relation is transported by the translation, modulo the syntactic equivalence;
and divergence reflection and success sensitiveness are conditions on the semantics.  
  
We adopt the first condition for our non-encodability results since we will require that the translation is compositional. The name invariance is irrelevant since our simple calculi do not have names. 
We do not use the third condition in the proposed form, since it has a flavour of showing equivalence of bisimulations, instead, we require equivalence of may- and must-convergence 
which is a bit weaker. Thus, for our non-encodability result the property could be included (still showing non-encodability), but for the correct translation 
in \cref{thm:translation-k-3-len-6}, we did not check the property.
Convergence equivalence for may- and must-convergence is our replacement of Gorla's divergence reflection and success sensitiveness.

Translations from synchronous to asynchronous communication  are investigated in the $\pi$-calculus \cite{Honda:1991,boudol:1992,Palamidessi03,Palamidessi:97}. Encodability results are obtained for the $\pi$-calculus without sums \cite{Honda:1991,boudol:1992},
while Palamidessi analyzed synchronous and asynchronous communication in the $\pi$-calculus with \emph{mixed sums} and non-encodability is the main result \cite{Palamidessi:97,Palamidessi03}.

A high-level encoding of synchronous communication into shared memory concurrency is
an encoding of CML-events in Concurrent Haskell using MVars \cite{Russell01,Chaudhuri09},
however a formal correctness proof for the translation is not provided. 


{\em Outline.}
 In \cref{sec:simple-languages} we introduce the process language $\PISIMPLE$ with synchronous communication and the process language $\LOCKSIMPLE$ with asynchronous locks.
After defining the correctness conditions on translations, we show that three locks (with a specific initialization) are sufficient 
 for a correct translation and we discuss variants of the target language.
 In particular, we show that changing blocking variants is equivalent to a modification of the initial store.
 In  \cref{section:k=1-impossible} it is shown that one lock in $\LOCKSIMPLE$ is insufficient for a correct translation
 and \cref{sec:general-props} exhibits certain general properties of correct translations which use two or more locks.
 \Cref{sec:k=2-refuting} contains the structuring into different blocking types of translations,  and proofs that there are no 
 correct translations for two locks and any initial store. 
 \Cref{section:conclusion} concludes the paper.
 For space reasons some proofs are omitted, but they can be found in the extended version of this paper \cite{schmidt-schauss-sabel:21-conclock-ext}.



\section{Languages for Concurrent Processes}\label{sec:simple-languages}
We define abstract and simple models for concurrent processes with synchronous communication and for concurrency with synchronizing shared memory.
The former model is a simplified variant of the $\pi$-calculus with a single global channel name and without replication or recursion, the latter can be seen as a variant where 
interprocess communication is replaced by binary semaphores.
Thereafter we define correct translations, prove correctness of a specific translation and consider variants of the target language. 

\subsection{The Calculus \PISIMPLE}
%
%
%
%
\begin{definition}
The syntax of processes and subprocesses of the calculus \PISIMPLE~is defined by the following grammar, where $i \in \{1,\ldots,k\}$:

\begin{center}
\begin{tabular}{ll@{~}c@{~}l}
Subprocesses &  $\mathcal{U}$ & $::=$ & $\success \mid \silent \mid\ !\mathcal{U} \mid\  ?\mathcal{U}$  \\

Processes &  
        ${\cal P}$ & $::=$ & $\mathcal{U} ~|~ \mathcal{U} \PAR {\cal P}$ 
\end{tabular}
\end{center}
\end{definition}
We informally describe the meaning of the symbols.
The symbol 
  $\silent$ means the silent subprocess; the symbol  
$\success$ means success,
     The operation $!$ means an output (or send-command), and $?$ means an input (or receive-command),
         and $\PAR$ is parallel composition.
For example, the expression $?!!\success \PAR !?\silent$  is a process, and so are  
also $???!!!?\success$  and  $?!!\success \PAR !?\silent \PAR \success \PAR !!!!!!?!\success$.
We assume that $\PAR$ is commutative and associative
and that $\silent$ is an identity element {w.r.t.}~$\PAR$,
{i.e.}~$\silent \PAR P = P$ for all $P$. 
Thus a process can be seen as a multiset of subprocesses.


%

\begin{definition} 
The {\em operational semantics} of $\PISIMPLE$ 
is a (non-deterministic) small-step operational semantics.
A single step
$\xrightarrow{\PIS}$
is defined as
    $$!\mathcal{U}_1 \PAR ?\mathcal{U}_2 \PAR {\cal P}   \xrightarrow{\PIS}  \mathcal{U}_1 \PAR \mathcal{U}_2 \PAR {\cal P}$$  
where $\mathcal{U}_1,\mathit{U_2}$ are arbitrary subprocesses and ${\cal P}$ is an arbitrary process.

The reflexive-transitive closure of  $\xrightarrow{\PIS}$ is denoted as  $\xrightarrow{\PIS,*}$.

If a process is of the form $\success\PAR{\cal P}$,
then the process is  {\em successful}. 
A sequence of $\xrightarrow{\PIS}$-steps starting with ${\cal P}$ is called \emph{an execution of ${\cal P}$}.
\end{definition}

Note that there may be several executions of processes, but every execution terminates.
\begin{example} Two examples for the execution of $P =\  ?!\silent \PAR  !!\success \PAR  ?\silent$ are: 
\begin{itemize}
\item $ P =  {?}{!}\silent  \PAR  !!\success \PAR  ?\silent    ~\xrightarrow{\PIS}~  !\silent  \PAR   !\success \PAR  ?\silent    
     ~\xrightarrow{\PIS}~  !\silent \PAR  \success \PAR  \silent$,
     where the final process is successful.
\item $P =  {?}!\silent \PAR  !!\success \PAR  ?\silent~\xrightarrow{\PIS}~  !\silent \PAR   !\success \PAR  ?\silent
    ~\xrightarrow{\PIS}~ \silent \PAR   !\success \PAR  \silent $
    where the final process is terminated, but not successful.
\end{itemize}
 This means there may be executions leading to a successful process, and at the same time executions leading to a fail.
  \end{example}
 We often omit the suffix $\silent$ for a subprocess, {i.e.}~whenever a subprocess ends with symbol $!$ or $?$ we mean the same subprocess extended by $\silent$.

 \begin{definition} A process $\mathcal{P}$ is called
 \begin{itemize}
   \item {\em may-convergent}  if there is some successful process ${\cal P}'$ with ${\cal P} \xrightarrow{\PIS,*} {\cal P}'$. 
  
  \item {\em must-convergent}  if for all processes ${\cal P}'$  with ${\cal P} \xrightarrow{\PIS,*} {\cal P}'$, the process ${\cal P}'$ is may-convergent. 
   \item {\em must-divergent} or a {\em fail}, if there is no execution leading to a successful process.
  \item  {\em may-divergent}, if there exists an execution ${\cal P} \xrightarrow{\PIS,*} {\cal P}'$, where ${\cal P}'$ is a fail.
 \end{itemize}  
\end{definition}

Our definition of must-convergence is the same as so-called should-convergence 
(see {e.g.}~\cite{sabel-schmidt-schauss-MSCS:08,schmidt-schauss-sabel-maymust:10,sabel-schmidt-schauss-PPDP:2011}). However, since there are no infinite reduction sequences, 
the notions of should- and must-convergence coincide (see {e.g.}~\cite{rensink-vogler:07,schmidt-schauss-sabel-maymust:10} for more discussion on the different notions). 
Thus, an alternative but equivalent definition of must-convergence is: a process $P$ is must-convergent, if all maximal reductions starting from $P$ end with a successful process.

%
 \subsection{The Calculus {\LOCKSIMPLE}}
 We now define the calculus $\LOCKSIMPLE$ which can be seen as a modification of $\PISIMPLE$ where ? and ! are removed, and operations $P_i$
 and $T_i$, which mean put and take, are added where $i=1,\ldots,k$ and $k$ is the number of locks ({i.e.}~storage cells). Locks can be empty (written as $\eempty$) or full (written as $\full$). For $k$ locks, the {\em initial store} is a $k$-tuple $(C_1,\ldots,C_k)$ where $C_i \in \{\eempty,\full\}$. We make this explicit by writing ${\LOCKSIMPLE}_{k,IS}$ for the language with $k$ locks and initial store $IS$.
 Subprocesses in ${\LOCKSIMPLE}_{k,IS}$ for a fixed value $1 \leq k \in \bbbn$ are
 built from $\success,\silent$, the symbols  $P_i,T_i$ and concatenation.
%
%
 Processes are a multiset of subprocesses:
 they are composed by parallel composition $\PAR$ which is assumed to be associative and commutative.
  
  \begin{definition} The syntax of processes and subprocesses of the calculus ${\LOCKSIMPLE_{k,IS}}$  is defined by the following grammar:
\begin{center}  
   \begin{tabular}{llcl}
     subprocess: &  $\mathcal{U}$ & $::=$ &   $\silent \mid \success \mid   P_i\mathcal{U} \mid T_i\mathcal{U}$
     \\
  process: & ${\cal P}$ & $::=$ &  $\mathcal{U} ~|~ \mathcal{U} \PAR {\cal P}$
   \end{tabular} 
\end{center}

  \end{definition}

We first describe the operational semantics of processes of {$\LOCKSIMPLE_{k,IS}$} and then give the formal definition.
The operational semantics is a non-deterministic small-step reduction $\xrightarrow{LS}$ which operates on 
$k$ locks $C_i$ (which are full ({i.e.}~$\full$)  or empty (written as $\eempty$)). The  execution of the operations $P_i$ or $T_i$ is as follows:
\begin{center}   
  \begin{tabular}{lll}
      $P_i$:& (put) &  changes $C_i$ from $\eempty \to \full$, or waits, if $C_i$ is $\full$.\\
      $T_i$:& (take) & changes $C_i$ from $\full \to \eempty$, or goes on (no change), if $C_i$ is $\eempty$
   \end{tabular}
\end{center}  
Note that locks together with $P_i$ and $T_i$ behave like binary semaphores, where $(P_i,T_i)$ means (wait,signal) (or (down,up), resp.). The semaphore is set to $1$ 
if the lock is empty, and set to $0$ if the lock is full. Note that locks specify a particular behavior for the 
case of a signal operation and the semaphore set to 1:
the signal has no effect (since $T_i$ on an empty lock does not have an effect).
Now we formally define the operational semantics:
 \begin{definition}
The relation $\xrightarrow{LS}$ operates on  a pair $({\cal P},(C_1, \ldots, C_k))$, where ${\cal P}$ is a $\LOCKSIMPLE_{k,IS}$-process, $C_1,\ldots,C_k$ are the storage cells. For a $\LOCKSIMPLE_{k,IS}$-process ${\cal P}$ the reduction starts with initial store $({\cal P},IS)$.
   
  We write the state as ${\cal C}$, and with ${\cal C}[C_i = \eempty]$ we denote that the specific cell $C_i$ has value $\eempty$. The notation 
  ${\cal C}[C_i \mapsto \eempty]$ means that in ${\cal C}$ the value in storage cell $C_i$ is replaced by $\eempty$. The same for $\full$ instead of $\eempty$. 
    The relation $\xrightarrow{LS}$ is defined by the following two rules:
   $$
  (P_i{\cal U} \PAR {\cal P},{\cal C}[C_i = \eempty]) \xrightarrow{LS}({\cal U} \PAR {\cal P},{\cal C}[C_i \mapsto \full])
   \quad\text{and}\quad
   (T_i{\cal U} \PAR {\cal P},{\cal C})
   \xrightarrow{LS}({\cal U} \PAR {\cal P},{\cal C}[C_i \mapsto \eempty])
   $$
  
%
%
 The reflexive-transitive closure of $\xrightarrow{LS}$ is denoted as $\xrightarrow{LS,*}$. 
 A sequence $(\mathcal{P},\mathcal{C}) \xrightarrow{LS,*} (\mathcal{P}',\mathcal{C}')$
 is called an \emph{execution of 
 $(\mathcal{P},\mathcal{C})$}, and if
 $\mathcal{C} =  IS$ then it is also called an \emph{execution of $\mathcal{P}$}.
 \end{definition}
 
To simplify notation, we write $\LOCKSIMPLE_k$ for the language with $k$ locks where all locks are empty at the beginning, {i.e.}~it is $\LOCKSIMPLE_{k,IS}$ with $IS = (\eempty,\ldots,\eempty))$.

Note that the blocking behavior of the put-operation is modelled by the operational semantics as follows:
for  $(P_i{\cal U} \PAR {\cal P},~{\cal C}[C_i = \full])$ there is no step (for subprocess $P_i{\cal U}$) defined and thus  $P_i{\cal U}$ has to wait until another subprocess changes the value of $C_i$.
  
\ignore{
 \begin{definition} A process $\mathcal{P}$ is called
 \begin{itemize}
   \item {\em successful}, if there is a subprocess $\success$ of $\mathcal{P}$, {i.e.}~$\mathcal{P} = \success \PAR \mathcal{P}'$ for some $\mathcal{P}'$. 
   \item {\em may-convergent},  if there is some successful process $\mathcal{P}'$ with $\mathcal{P} \xrightarrow{LS,*} \mathcal{P}'$. 
  
  \item {\em must-convergent},  if for all processes $\mathcal{P}'$  with $\mathcal{P} \xrightarrow{LS,*} \mathcal{P}'$, the process $\mathcal{P}'$ is may-convergent. 
   \item {\em must-divergent} or a {\em fail}, if there is no execution leading to a successful process.
  \item  {\em may-divergent},  if for some process $\mathcal{P}'$:  $\mathcal{P} \xrightarrow{LS,*} \mathcal{P}'$, where $\mathcal{P}'$ is a fail.
 \end{itemize}  
 \end{definition}
}
\begin{definition}   
A process $\mathcal{P}$ of $\LOCKSIMPLE_{k,IS}$ is called
 {\em successful}, if there is a subprocess $\success$ of $\mathcal{P}$, {i.e.}~$\mathcal{P} = \success \PAR \mathcal{P}'$ for some $\mathcal{P}'$.
%
A state $(\mathcal{P},{\cal C})$ is called 
 \begin{itemize}
   \item {\em successful}, if $\mathcal{P}$ is successful.
   \item {\em may-convergent},  if there is some successful $(\mathcal{P}',{\cal C}')$ with $(\mathcal{P},{\cal C}) \xrightarrow{LS,*} (\mathcal{P}',{\cal C}')$. 
  
  \item {\em must-convergent},  if for all states $(\mathcal{P}',{\cal C}')$  with $(\mathcal{P},{\cal C}) \xrightarrow{LS,*} (\mathcal{P}',{\cal C}')$, the state
   $(\mathcal{P}',{\cal C}')$ is may-convergent. 
   \item {\em must-divergent} or a {\em fail}, if there is no execution leading to a successful state.
  \item  {\em may-divergent},  if for some state $(\mathcal{P}',{\cal C}')$:  $(\mathcal{P},{\cal C}) \xrightarrow{LS,*} (\mathcal{P}',{\cal C}')$, where $(\mathcal{P}',{\cal C}')$ is a fail.
\end{itemize}
A process $\mathcal{P}$ is called {\em may-convergent}, {\em must-convergent}, {\em must-divergent}, or {\em may-divergent}, resp.~iff the state $(\mathcal{P},IS)$ is 
{\em may-convergent}, {\em must-convergent}, {\em must-divergent}, or {\em may-divergent}, resp.
%

 \end{definition}
 %
An example for a reduction sequence for $k = 2$ is: 
%
 $$ (P_2\silent \PAR T_2\success, (\eempty,\eempty))  
     \xrightarrow{LS}    (\silent \PAR T_2\success,  (\eempty, \full)) 
     \xrightarrow{LS}   (\silent \PAR \success,  (\eempty, \eempty)) 
      ~~~  (\text{successful}) 
  $$    
%
The process  $P_2\silent \PAR T_2\success$ is  even must-convergent.
  
In the following, we often leave the state implicit and in abuse of notation, we 
``reduce'' processes without explicitly mentioning the state.

As in $\PISIMPLE$ we often omit the suffix, $\silent$,  for a subprocess, {i.e.}~whenever a subprocess ends with symbol $P_i$ or $T_i$ we mean the same subprocess extended by $\silent$.  

%
  

 \subsection{Correct Translations}\label{subsec:correct-translations}
 
%

We are interested in translations from
 one full  concurrent programming language with synchronous semantics into  another full imperative concurrent language with locks, 
 where the issues are expressive power and the comparison between the languages.
In order to focus considerations, we investigate this issue by considering translations from a core concurrent language (SYNCSIMPLE) with 
synchronous semantics into a core of an imperative concurrent language (LOCKSIMPLE).   

%
%
%
%
 However, even in our simple languages there are
 interesting questions, for example, whether there exists
 a correct translation 
 and how many locks are necessary for such a translation. 

 Since our analysis started top-down, we are sure that  the non-encodability results can be transferred back to larger calculi.
For discussing this, let us call the full languages $\mathrm{SYNCFULL}$ and $\mathrm{LOCKFULL}$, resp. The language $\mathrm{SYNCFULL}$ may be  the $\pi$-calculus and thus, it extends $\PISIMPLE$ by names, named channels, name restriction, sending and receiving names and replication or recursion. The language $\mathrm{LOCKFULL}$ 
may be a variant of the core language of Concurrent Haskell, where locks are extended to synchronising memory cells which have addresses (or names) and content (for instance, numbers).
The main argument why non-encodability in the small languages implies
non-encodability in the larger languages is the following: 
Suppose we have non-encodability between the small languages for $2$ locks, and  
 there  exists a correct (compositional) translation $\phi:\mathrm{SYNCFULL}\to\mathrm{LOCKFULL}$ that uses only one synchronising memory cell in $\mathrm{LOCKFULL}$.
 Then the  idea is to embed every $\PISIMPLE$-program $\mathcal{P}$ into a $\mathrm{SYNCFULL}$-program $\mathcal{P}'$
 by using only one channel, and then using the translation $\phi$ to derive a $\mathrm{LOCKFULL}$-program $\phi(\mathcal{P}')$.  
 Using this construction, we also get a translation of $!$ and $?$ into $\mathrm{LOCKFULL}$, where every ! translates into a send-prefix, and every ? into a receive-prefix. 
 The parallel-operator remains as it is. Then the correctness of $\phi$ tells us that the $\mathrm{LOCKFULL}$-program
$\phi(\mathcal{P}')$ has the same may- and must-convergencies. 
 Compositionality gives us a $\LOCKSIMPLE$-program that uses at most $2$ locks, 
  and  it has the same parallel-structure as $\mathcal{P}$, and the !,?, are translated always in the same way.  
 The result can be reduced to a $\LOCKSIMPLE$-program with at most $2$ locks, (perhaps after restricting $\phi$ \wrt~contents of messages and recursion), 
 which  contradicts the result on  small languages,
 since the reasoning holds for all $\mathcal{P}$.

 \begin{definition}
  A mapping $\tau$ from the processes of {\PISIMPLE} into processes of $\LOCKSIMPLE_{k,IS}$ is called a {\em translation}.
  \begin{itemize}
    \item $\tau$ is called {\em compositional}  iff 
         $\tau(\silent) = \silent$,  $\tau(\success) = \success$, 
         $\tau(\mathcal{P}_1 \PAR \mathcal{P}_2) = \tau(\mathcal{P}_1) \PAR \tau(\mathcal{P}_2)$; $\tau(\mathcal{U})$ does not contain the parallel operator $\PAR$ for every subprocess $\mathcal{U}$; and  
         $\tau(!\mathcal{U}) = \tau(!)\tau(\mathcal{U})$ and $\tau(?\mathcal{U}) = \tau(?)\tau(\mathcal{U})$ for every subprocess $\mathcal{U}$ 
    \item $\tau$ is called {\em correct} iff for all  \PISIMPLE-processes $P$, $P$ is may-convergent iff $\tau(P)$ is may-convergent,
       and   $P$ is must-convergent iff $\tau(P)$ is must-convergent,
  \end{itemize}
Compositional translations $\tau$ in our languages
can be identified  with the pair $(\tau(!),\tau(?))$ of strings, and we say that
$\tau$ has \emph{length} $n$, if $|\tau(!)| + |\tau(?)| = n$. 
 \end{definition}
 
 For example, a correct translation cannot map $\tau(\silent)  = \success$ since then $\silent$ is must-divergent, but $\tau(\silent)$ is must-convergent.
 Hence  $\tau(\silent)  = \silent$  and $\tau(\success)  = \success$ make sense  for correct translations.
  
%

 We show that three locks are sufficient for a correct compositional translation.
 
 \begin{theorem} \label{thm:translation-k-3-len-6}
    For $k = 3$, the translation $\tau$ with  
    $\tau(!) =  P_1T_3P_2T_1$ and  $\tau(?) = P_3T_2$   is correct for initial store $(\eempty,\full,\full)$.    
 \end{theorem}
 \begin{proof}
We give a sketch (the full proof can be found in \cite{schmidt-schauss-sabel:21-conclock-ext}): 
 A communication starts with executing $P_1$ of  $\tau(!) =  P_1T_3P_2T_1$, leaving the storage $(\full,\full,\full)$. 
 Then no other  sequence $\tau(!),\tau(?)$ in parallel processes can be executed.  Then $T_3$ is executed,  leaving the storage $(\full,\full,\eempty)$.  
 The next step is  that one process with  $\tau(?) = P_3T_2$  may start, and   $P_3$ is executed,  leaving the storage $(\full,\full,\full)$. 
 Now $T_2$ is executed, and this is the only possibility. the storage  is then $(\full,\eempty,\full)$.   Again, the only possibility is now $P_2$ from $\tau(!)$ 
  and the storage $(\full,\full,\full)$. 
 The last step is executing  $T_1$, which restores the initial storage $(\eempty,\full,\full)$. 
 
 This is the only execution possibility of  $\tau(!)$ and $\tau(?)$, hence it can be retranslated into an interaction communication of a single $!$ and a single $?$.
  \end{proof}
  
  There are also other correct compositional translations for $k = 3$: 
  An example is a compositional correct translation $\tau$ of length $8$, detected by an automated search,
  with
  $\tau(!) =  P_2P_1T_3P_1T_1T_2$ and  $\tau(?) = P_3T_1$   and with initial store $(\eempty,\eempty,\full)$.    
 
 The observation is that  the communication is completely protected by using $P_2$ as a mutex, which is similar to the translation of length 6 
   (see Theorem \ref{thm:translation-k-3-len-6})
       
 \subsection{\texorpdfstring{Blocking Variants of $\LOCKSIMPLE$}{Blocking Variants of LOCKSIMPLE}}
We choose for our locks, that $P_i$ blocks, but $T_i$ never blocks.
However, also other choices are possible.
Variants of $\LOCKSIMPLE$ where for every $i$ either $P_i$ blocks 
on a full lock, but $T_i$ is non-blocking, or $T_i$ blocks on an empty lock, but $P_i$ is non-blocking, do not lead to really new problems:
  In  \cite{schmidt-schauss-sabel:21-conclock-ext} we show that all those variants are equivalent to the previously defined language where for all $i$: 
 $P_i$ is blocking, but $T_i$ is non-blocking. 
 This is possible since we take into account any initial store and thus the main argument of the equivalence is that we can change 
 the initial store for every $i$ by switching the role of $P_i, T_i$ and at the same time switching the initial store
 for $i$ from $\full$ to $\eempty$ and vice versa.
 Thus this extension does not increase the number of (really) different languages for a fixed $k$. 
%
However, the variant where $P_i$ blocks for a full lock \emph{and} $T_i$ blocks for an empty lock for all $i$ (which is related to an implementation using the MVars in Concurrent Haskell) appears to be different from our $\LOCKSIMPLE$ languages.
 There are results on possibility and impossibility  of correct translations from $\PISIMPLE$ into  a further restricted variant of $\LOCKSIMPLE$  
 \cite{schmidt-schauss-sabel-wpte:2020}. A deeper investigation in these languages is future work.  
%
%
 
 \section{One Lock is Insufficient for any Initialization}\label{section:k=1-impossible}
  We show that there is no correct (compositional) translation into $\LOCKSIMPLE_{1,IS}$, the language with one lock,  for any initial storage, {i.e.}~for initial storage $\full$ and initial storage $\eempty$.
 
 \begin{lemma}\label{lemma:storage-1lock-PP}
   Let $\tau$ be a  correct  translation $\PISIMPLE~\to \LOCKSIMPLE_{1,\mathit{IS}}$.
   Then $\tau(!)$ as well as $\tau(?)$ either start  with $P_1$ or have  a subsequence $P_1P_1$.  
  \end{lemma}  
  \begin{proof}
      Consider the processes $!\success$ and  $?\success$  which are both must-divergent. 
      If $\tau(!)$ does not satisfy the condition, then the process $\tau(!\success)$ can be executed without any wait
       and is successful.   
      The same for $\tau(?\success)$. However, this is a contradiction to correctness.
  \end{proof}


\begin{theorem}\label{thm:storage-1lock-impossible}
There is no correct translation  $\PISIMPLE~\to \LOCKSIMPLE_{1,IS}$.
\end{theorem}
\begin{proof}
Let $\tau$ be a correct translation. 
We first consider the case that the initial storage is $\eempty$. 
Then from \cref{lemma:storage-1lock-PP} we derive that $\tau(!)$ as well as $\tau(?)$ have  a subsequence $P_1P_1$ or start with $P_1$. 
since $P_1$ as a prefix is executable (and similar as in the proof of \cref{lemma:storage-1lock-PP},  the processes $!\success$ and $?\success$ can be used as examples to refute the correctness of $\tau$).
Consider the process $\tau(!\success \PAR ?\success)$, which is must-convergent. 
First, reduce $\tau(!\success)$ until exactly before the first occurrence of $P_1P_1$. 
Then reduce $\tau(?\success)$. Since the reduction starts with  $C_1 = \eempty$, it will block after executing the first $P_1$ of the leftmost subsequence $P_1P_1$ (or earlier). Then $C_1 = \full$, and we have a deadlock. This is a contradiction to correctness of $\tau$.
  
Now we consider the case that the initial store is $\full$.
Then \cref{lemma:storage-1lock-PP} shows that $\tau(!)$ and $\tau(?)$ contain a subsequence $P_1P_1$ or start with $P_1$.
We again use the must-convergent process $\tau(!\success \PAR ?\success)$.
If both $\tau(!)$ and $\tau(?)$ start with $P_1$, then there is an initial deadlock.
Suppose that neither $\tau(!)$ nor $\tau(?)$  do start with a $P_1$, then they both start with a $T_1$, 
and have a subsequence $P_1P_1$. Let us consider the leftmost such subsequence for $\tau(!)$ 
as well as for $\tau(?)$. Construct the following execution for $\tau(!\success \PAR ?\success)$: 
First $\tau(!)$ until it blocks at the second $P_1$ of the sequence $P_1P_1$, then the 
execution of $\tau(?)$ until the second $P_1$ of the sequence $P_1P_1$. Then we have a deadlock, which is impossible.
 
If $\tau(!)$ starts with a $P_1$, but not $\tau(?)$, then there is a leftmost sequence $P_1P_1$ of $\tau(?)$. 
Execute $\tau(?)$ until it is blocked at $P_1$. Then we reach a deadlock.  This is a contradiction. 
\end{proof}

 \section{General Properties for at Least Two Locks}\label{sec:general-props}  
 
In this section, we consider compositional translations  $\PISIMPLE~\to \LOCKSIMPLE_{k,IS}$ with $k \geq 2$ and prove several properties of
correct compositional translations that will help us later to show that
$k = 2$ is impossible.
We also introduce the notion of a blocking type for a translation. The idea of this notion is recording how $\tau$ establishes that executing $\tau(!)$ in the process $\tau(!\success)$
blocks and why executing $\tau(?)$ in the process $\tau(?\success)$ blocks.
Both processes must block if $\tau$ is correct, since the 
the processes $!\success$ and $?\success$ are both blocking (and not successful) in 
$\PISIMPLE$.

Below this notion helps to structure the arguments for different cases.
 
 \begin{lemma}\label{lemma:reduction-uses-all-symbols} Let $\tau$ be a correct translation from   $\PISIMPLE~\to \LOCKSIMPLE_{k,IS}$  for $k \geq 1$. 
 Then there is a reduction sequence of 
  $\tau(!) \PAR \tau(?)$ that executes every symbol in $\tau(!) \PAR \tau(?)$.
 \end{lemma}
 \begin{proof}  
 First, consider $\tau(!\success) \PAR \tau(?\silent)$, which is must-convergent (since $!\success \PAR ?\silent$ is must-convergent), and hence there is a reduction sequence of $\tau(!) \PAR \tau(?)$ consuming 
 at least all symbols in $\tau(!)$. The same sequence can be used as a partial reduction sequence of $\tau(!\silent) \PAR \tau(?\success)$, and since  this process is must-convergent (since $!\silent \PAR?\success$ is must-convergent), 
 the sequence will also consume all symbols of $\tau(?\success)$.
 \end{proof}
 
 
 The notation $\#(S,r)$ means the number of occurrences of the symbol $S$ in the string $r$.
 \begin{proposition}\label{prop:estimation-P-leq-T}
 Let $\tau:\PISIMPLE~\to \LOCKSIMPLE_{k,IS}$ for $k \geq 2$ be a correct translation. 
 Then for every $1 \leq i \leq k$:  $\#(P_i,\tau(!)) + \#(P_i,\tau(?)) \leq \#(T_i,\tau(!)) + \#(T_i,\tau(?))$. 
 \end{proposition}
\begin{proof}
  The processes $!!\success \PAR ??\success$,  $!!\silent \PAR ??\success$  and $!!\success \PAR ??\silent$ are must-convergent, hence also their images under $\tau$.
  Now suppose the claim is false. Then for some index, say 1, $\#(P_1,\tau(!)) + \#(P_1,\tau(?)) > \#(T_1,\tau(!)) + \#(T_1,\tau(?))$. 
We apply Lemma \ref{lemma:reduction-uses-all-symbols} to $\tau(!!\success \PAR ??\success)$ and obtain a reduction sequence $R_1$ that exactly consumes the 
  top parts $\tau(!)$ and $\tau(?)$  of  $\tau(!!\success \PAR ??\success)$. 
  


Replacing $\success$ by $\silent$, the reduction sequence $R_1$ can be also used 
for $\tau(!!\success \PAR ??\silent)$.
Since $\tau(!!\success \PAR ??\silent)$ is must-convergent, $R_1$ can be continued to $R_1R_2$ ending in a success of the form $\success \PAR Q\silent$ 
where $Q$ is a suffix of $\tau(?)$, since  $!!\success \PAR ??\silent$ is must-convergent.

The reduction sequence $R_1R_2$ can also be used for $\tau(!!\silent \PAR ??\success)$ (by interchanging $\silent$ and $\success$), ending in $\silent \PAR Q\success$. Since $!!\silent \PAR ??\success$ is must-convergent, the reduction sequence $R_1R_2$ can be extended to $R_1R_2R_3$ resulting in $\silent \PAR \success$.

   After $R_1$, we have $C_1 = \full$ and that the initial store for index $1$ is $\eempty$, due to the assumption, and since the symbols in $\tau(!),\tau(?)$ are completely 
   consumed.
   Hence $R_2R_3$ must execute  a $T_1$ before  every other $P_1$.
   But since the number of $T_1$-symbols is strictly smaller than the number of $P_1$-symbols, there must be a deadlock situation at least
    for one of the symbols $P_1$. 
   
   This is a contradiction, hence the proposition holds.
 \end{proof}
 
%
%

\begin{definition} \label{def:blocking-types}
For a correct translation $\tau$ into  $\LOCKSIMPLE_{k,IS}$, 
a {\em blocking prefix} of a sequence $S$ of symbols  in $\LOCKSIMPLE_{k,IS}$ is a prefix of $S$ of one of the two forms:
\begin{enumerate}
   \item   $R_1P_iR_2P_i$, 
       where $R_1,R_2$ are sequences, and $R_2$ does not contain $P_i,T_i$, and the execution of  $S$ that starts with store $IS$ deadlocks exactly before 
        the last symbol, which is $P_i$.
    \item $R_1P_i$, where $R_1$  does not contain $P_i,T_i$, and the execution of $S$ that starts with store $IS$  deadlocks exactly before the last symbol, which is $P_i$.
\end{enumerate}
    We may also speak of $R_1P_i$ or $P_iR_2P_i$, respectively, as a {\em blocking subsequence} of $S$.\\
     In the case that $S$ has a blocking sequence, we say that the {\em blocking type} of $S$ is $P_iP_i$ if the blocking sequence is $R_1P_iR_2P_i$, and the blocking type is $P_i$, 
     if the blocking sequence is $R_1P_i$. 
  
  We say a translation $\tau$ has {\em blocking type} $(W_1, W_2)$, if $W_1$ is the blocking type of $\tau(!)$, and $W_2$ is the blocking type of $\tau(?)$.
  \end{definition}

 \begin{lemma}\label{lemma:P1P1-blocking-exists}
 Let $\tau:$ $\PISIMPLE~\to \LOCKSIMPLE_{k,IS}$ be a correct translation  where $k \geq 2$. 
 Then there is some $i$, such that $\tau(!)$ has a blocking subsequence of the form $RP_i$, or $P_iRP_i$, 
    where $R$ does not contain $P_i,T_i$. The same holds for $\tau(?)$. 
 \end{lemma}
\begin{proof}
 The reduction of $\tau(!)$ cannot be completely executed, since $\tau(!\success)$ is a fail. Hence  the execution stops
 at a symbol $P_i$, and it is either the first occurrence of $P_i$, or a later occurrence. Hence the sequence before is of the form $R$, or $P_iR$, where $R$ does not contain $P_i,T_i$.
 The same arguments hold for $\tau(?)$.
\end{proof}

\begin{lemma}\label{lemma:P1-type-and-store}
 Let $\tau:$ $\PISIMPLE~\to \LOCKSIMPLE_{k,IS}$ be a correct translation  where $k \geq 2$. 
 If $\tau(!)$ is of  blocking type $P_i$ then $IS_i = \full$, and 
 if $\tau(!)$ is of  blocking type $P_iP_i$ then the first $i$-symbol is $T_i$, or $IS_i = \eempty$;
 The same holds for $\tau(?)$.
  \end{lemma}
\begin{proof}
   The blocking type $P_i$ is only possible if in $R$ of the prefix $RP_i$ there is no $T_i$, hence the initial store $\mathit{IS}_i = \full$.
    If the blocking type is $P_iP_i$ and  $IS_i = \full$, then the first $i$-symbol must be a $T_i$. The other case is that $\mathit{IS}_i$ is $\eempty$.
  %
\end{proof}
  
%
%
%
%
%

\section{Non-Existence of a Correct Translation for Two Locks}\label{sec:k=2-refuting}    
In this section, we will show that there is no correct compositional translation
from $\PISIMPLE$ to $\LOCKSIMPLE_{2,IS}$  (for any initial storage $IS$). We distinguish several cases by considering different blocking types according to 
\cref{def:blocking-types}. 
When reasoning on translations, 
we use an extended notation of translations
as pairs of strings (i.e.~$(\tau(!),\tau(?))$): We describe sets of translations using set-concatenation (writing singletons without curly braces) and the Kleene-star.
For instance, we write $(\{P_1,T_1\}^*T_2,\{P_2\}^*T_1)$ to denote the set of all translations where $\tau(!)$ starts with arbitrary many $P_1$- and $T_1$-steps ending with $T_2$, and $\tau(?)$ starting with an arbitrary number of $P_2$-steps followed by a single $T_1$-step.


An automated search for compositional translations for $k = 2$ and length $\le 10$ has refuted the correctness of all these translations for all initializations of the initial storage. This is consistent with our general arguments in this section.
\subsection{\texorpdfstring{Refuting the Blocking Type $(P_iP_i,P_jP_j)$}{Refuting the Blocking Type (PiPi,PjPj)}}
 \begin{proposition}\label{proposition:k-is-2-and-i-not-j}   
 Let $\tau:$ $\PISIMPLE~\to \LOCKSIMPLE_{2,IS}$ be a correct translation of blocking type $(P_iP_i, P_jP_j)$.
  Then  $i \not= j$. 
 \end{proposition}
\begin{proof}
W.l.o.g.~assume that the blocking type is $(P_1P_1, P_1P_1)$.
Then the blocking prefixes  of $\tau(!)$ and  $\tau(?)$ are   $M_1P_1RP_1$ and $M_2P_1R'P_1$, respectively. 
Now
we reduce the must-convergent process $\tau(!\silent) \PAR \tau(?\success)$ by selecting the following reduction sequence: 
  first, reduce $\tau(!)$ until $M_1P_1R$ is completely executed, and then
 reduce $\tau(?)$ as far as possible. 
 Let $Q$ be the prefix of $\tau(?)$  of the form $\{P_2,T_2\}^*\{P_1,T_1\}$.  
 If $P_1$ is the  symbol from $\{P_1,T_1\}$ of $\tau(?)$, then a deadlock would occur, which is not possible, since  $!\silent \PAR ?\success$ is must-convergent.
 Hence $Q$ as a prefix of $\tau(?)$ must be of the form   $\{P_2,T_2\}^*T_1$. 
%
 There are two cases:
 \begin{enumerate}
   \item  After executing $M_1P_1R$ it holds $C_1 = \full$ and $C_2 = \eempty$.  
     \begin{enumerate}
       \item  $IS_2 = \eempty$.  
        Now, since reducing $\tau(?)$ starts with $C_2 = \eempty$, and the final $T_1$ of $Q$  resets $C_1$, the reduction sequence starting with $\tau(!\silent)$ and then executing  $\tau(?)$  is possible 
         until the end of $M_2P_1R'$. 
         Since now $C_1 = \full$ and in both pending subprocesses a $P_1$ is to be executed, we have a  deadlock, which is impossible due to must-convergence of $!\silent \PAR ?\success$.
         \item   $IS_2 = \full$.  Then the first $\{P_2,T_2\}$-symbol of $\tau(?)$ cannot be $P_2$, since it would block. 
            Hence the first $\{P_2,T_2\}$-symbol of $\tau(?)$ is $T_2$. Then the further reduction of $\tau(?)$ is independent of the initial values and it is the same as 
            in the previous case.

       \end{enumerate}
   \item After executing $M_1P_1R$ it holds $C_1 = \full$ and $C_2 = \full$.   
      Then the first symbol of $\tau(?)$ cannot be $P_2$, since this would be a deadlock.
      Also, the first symbol of $\tau(?)$ cannot be $T_2$, since then the reduction of  $\tau(?)$  alone is the same as started with the initialization $C_1 = C_2 = \eempty$,
      and the reduction proceeds until the end of the blocking sequence, which leads to a deadlock. 
      Hence $\tau(?)$ starts with $T_1$. 
      The prefix of $\tau(?)$ cannot be $T_1\{T_1,P_1\}^*P_2$, since this either blocks  within $T_1\{T_1,P_1\}^*$  or at $P_2$. 
      Hence the prefix is $T_1\{T_1,P_1\}^*T_2$. This implies that  $\tau(?)$ is executable until the blocking $P_1$, and thus leads to a deadlock. 
      Hence this case  is also not possible.
 \end{enumerate}
  We have checked all cases, hence $i = j$ is not possible and the lemma is proved.
\end{proof}

 We consider the  blocking type $(P_1P_1,P_2P_2)$ in the rest of this subsection,
which suffices due to symmetry and \cref{proposition:k-is-2-and-i-not-j}.
 
 \begin{lemma}\label{lemma:blocking-pre-B-all-inital}
  Let $\tau:$ $\PISIMPLE~\to \LOCKSIMPLE_{2,IS}$ be a correct translation of blocking type $(P_1P_1,P_2P_2)$.   
  Then the following holds:
 \begin{enumerate}
     \item The blocking prefix of $\tau(!)$ is  $R_1P_1\{P_2,T_2\}^*T_2P_1$ and  the blocking prefix of $\tau(?)$ is
       $R_3P_2\{P_1,T_1\}^*T_1P_2$.  
     \item  $\{T_1,P_1\}^*T_2$ is a prefix of $\tau(!)$, and $\{T_2,P_2\}^*T_1$ is a prefix of $\tau(?)$.
  \end{enumerate}
 \end{lemma}
 \begin{proof} 
 Let the blocking prefix of $\tau(!)$ be   $R_1P_1P_1$   and the blocking prefix of $\tau(?)$ be   $R_3P_2\{T_1,P_1\}^*P_2$. 
 Then first execute  $R_1$, and then $R_3P_2\{T_1,P_1\}^*$ until it blocks. 
    If it blocks at a $P_1$, then it is a deadlock. If it blocks at a $P_2$, then $P_1P_1$ cannot be both executed, hence a deadlock. 
 Hence $\tau(!)$ has a blocking prefix  $R_1P_1R_2P_1$ where $R_2 \not= \emptyset$.  
 By symmetry, we obtain
 that the blocking prefix of $\tau(?)$ is
 $R_3P_2R_4P_2$ where 
 $R_4 \not = \emptyset$.  
 Now let the blocking prefix of $\tau(!)$ be   $R_1P_1\{T_2,P_2\}^*P_2P_1$.
 Execute $\tau(!)$ until  $P_2P_1$ is left, and then execute  $\tau(?)$. Clearly, $\tau(?)$ must block, 
 independent of the previous executions.
 If $\tau(?)$ blocks at $P_1$, then we have a deadlock, and if it blocks at $P_2$, 
 then we also have a deadlock. 
 Hence the blocking prefix of $\tau(!)$ is of the form  $R_1P_1\{T_2,P_2\}^*T_2P_1$. 
 
 By symmetry, we obtain that the blocking prefix of $\tau(?)$ is of the form  $R_3P_2\{T_1,P_1\}^*T_1P_2$.  
 Now we prove restrictions on the prefix of $\tau(!)$ and $\tau(?)$. 
   Assume that the prefix of $\tau(?)$ is $\{T_2,P_2\}^*P_1$.  Then 
   first reduce
      $\tau(!)$ until it blocks before $P_1$, then reduce $\tau(?)$, until it blocks within $\{T_2,P_2\}^*$ or at the (first) $P_1$ in  $\tau(?)$.
      Both cases lead to a deadlock, 
          hence this case is impossible. 
       Thus $\tau(?)$ has prefix  $\{T_2,P_2\}^*T_1$. 
    \end{proof}
  
  \noindent For the rest of this subsection, we assume blocking type $(P_1P_1,P_2P_2)$, and that only correct translations are of interest.
 
%
%
%
%
\begin{lemma}\label{lemma:k-is-2-send-is-not-t1plus-t2}
Let $\tau$ be a correct translation. Then for any initial storage    
the prefix of $\tau(!)$ cannot be  $T_1^+T_2$ nor $T_2^+T_1$.
\end{lemma}
\begin{proof}
In each case the must-divergent process $\tau(!\success \PAR \ldots \PAR !\success)$ with sufficiently many subprocesses can be reduced such that it leads to a success, 
which contradicts the correctness of $\tau$:
    Fix the first subprocess and reduce it until the end using the prefixes of the other subprocesses to proceed in case of a blocking.
    This leads to success, which is a contradiction.
\end{proof}
 
 
 \cref{lemma:blocking-pre-B-all-inital} implies:
 \begin{lemma}\label{lemma:k-is-2-send-is-not-t1startp2}
The prefix of $\tau(!)$ cannot be  $T_1^*P_2$.
\end{lemma}

 \begin{lemma}\label{lemma:k2-T2+P2}
 Let $\tau$ be a correct translation.  
Then the prefix of $\tau(!)$ cannot be $T_2^+P_2$. 
\end{lemma}
\begin{proof}
Consider the must-convergent process $\tau(!\success \PAR \ldots \PAR !\success \PAR ?\silent)$. First reduce all the prefixes $T_2^+$ in all $\tau(!\success)$
 until $P_2$ is the first symbol. 
Since $\{T_2,P_2\}^*T_1$ is a prefix of $\tau(?)$, and due to the assumption of the blocking type, reduction cannot block at a $P_2$ in $\{T_2,P_2\}^*$.
Hence $T_1$ is executed, which means that reduction is now independent of the initial store.
 We reduce $\tau(?)$ until it stops before the second $P_2$ of the blocking subsequence.
Then it is  a deadlock, which contradicts correctness of $\tau$.
 \end{proof}

 \begin{lemma}\label{lemma:k2-excl-pref-not-P1}
The prefix of $\tau(!)$ cannot be  $P_1$.
\end{lemma}
\begin{proof} Assume the prefix of $\tau(!)$ is $P_1$. Then $IS_1 = \eempty$ due to the assumption that the blocking type is $(P_1P_1,P_2P_2)$.   
   Consider the must-convergent process $!\success \PAR \ldots \PAR !\success \PAR ?\silent$, where we fix the  number of  $!\success$-subprocesses later if this is necessary.
   We will use the structure of the subprocesses $\tau(!)$ and $\tau(?)$ proved in Lemma \ref{lemma:blocking-pre-B-all-inital} whenever necessary. 
   \begin{enumerate}
     \item  Reduce $\tau(?\silent)$ before it stops at the second $P_2$ of the blocking subsequence. After this we have $C_1=\eempty, C_2=\full$. 
     \item Reduce one subprocess $\tau(!\success)$ until it blocks. Since $C_1 = IS_1 = \eempty$ at the start and $\{P_1,T_1\}^*T_2$ is a prefix of $\tau(!)$, 
      the reduction is the same as started with $IS$, hence it stops at the second $P_1$ of the blocking subsequence and so $C_1 = \full, C_2 = \eempty$ at the end.
      \item We go on with the reduction of $\tau(?)$  until it blocks. It cannot block at a $P_1$, since this would be a deadlock.
       If the reduction consumes all of $\tau(?)$, then we reduce the next $\tau(!)$: The prefix $\{T_1,P_1\}^*T_2$ shows that it cannot block at $P_1$ of     $\{T_1,P_1\}^*$,
       since this would be a deadlock, hence $T_2$ is executed, Now it cannot block at a $P_2$ before the end of the blocking sequence. Thus reduction will lead to a deadlock at the 
       end of the blocking sequence, since all remaining subprocesses start with a $P_1$.
       
       The last case is that the further reduction of $\tau(?\silent)$ blocks at a $P_2$. Then again we reduce the next subprocess $\tau(!\success)$.
       It cannot block at $P_1$ of the prefix $\{T_1,P_1\}^*T_2$, since this would be a deadlock, hence it executes a $T_2$, and thus again it blocks at a $P_1$ at the end of a 
       blocking sequence. This is the final deadlock.\qedhere
   \end{enumerate}
\end{proof}


\begin{lemma}\label{lemma:k-is-2-send-is-not-t2plus-p1}
 Let $\tau$ be a correct translation.  
Then the prefix of $\tau(!)$ cannot be $T_2^+P_1$.
\end{lemma}
\begin{proof}
Consider the must-convergent process $\tau(!\success \PAR \ldots \PAR !\success \PAR ?\silent)$.
First, reduce all $T_2^+$-prefixes away, thhen use the same arguments as in  Lemma \ref{lemma:k2-excl-pref-not-P1}, which is possible, since it is the same process.
\end{proof}

\ignore{
\begin{proof}
Consider the must-convergent process $\tau(!\success \PAR \ldots \PAR !\success \PAR ?\silent)$.
First, reduce all $T_2^+$-prefixes away.  
 The same arguments as in the  proof of Lemma \ref{lemma:k2-T2+P2} show that the reduction of $\tau(?)$ is independent of the initial store.
 
\dstodo{Der $T_2^*P_1$-Fall enthält aber auch,
dass es keine $T_2$ gibt? Wieso ist dann klar, dass
$\tau(?)$ bis zum zweiten $P_2$ für jeden IS reduziert werden kann?}

 Reduce $\tau(?)$ until it stops before the second $P_2$ of the blocking subsequence.
 Now $C_1 = \eempty$, and since $\{T_1,P_1\}^*T_2$ is a prefix of $\tau(!)$ (Lemma \ref{lemma:blocking-pre-B-all-inital}), the reduction of $\tau(!)$ is  independent of the initial value (of $C_2$).
  Then reduce the first  $\tau(!\success)$ until the but last symbol $T_2$ of the blocking subsequence will be reduced (see Lemma \ref{lemma:blocking-pre-B-all-inital}). 
 Repeating this, using other subprocesses of the form $\tau(!)$, the subprocess $\tau(?)$ can be stepwise reduced
 until either it is completely reduced, or a symbol $P_1$ has to be reduced but $C_1 = \full$. \\
 In the first case reduce a further $\tau(!)$ until the end of the blocking sequence and hence $C_1 = \full$.

 \dstodo{Ist das klar, dass das immer geht? Eigentlich ja nur für den initial state, der aber nich unbedingt gegeben ist???}
 
 We can assume that there are no more subprocesses $\tau(!\success)$.
 Now all processes have $P_1$ as a first symbol and $C_1 = \full$, which is a deadlock.
 \end{proof}
 }
 
 Since $P_1$ as prefix of $\tau(!)$ is already excluded, we show the following.

\begin{lemma}\label{lemma:k-is-2-send-is-not-t1plus-p1}
Let $\tau$ be a correct translation.  
Then the prefix of $\tau(!)$ cannot be $T_1^+P_1$.
\end{lemma}
\begin{proof} Let us assume that the prefix of $\tau(!)$ is $T_1^+P_1$. We know that it is also $\{P_1,T_1\}^*T_2$. Consider the must-convergent process $\tau(!\success \PAR \ldots \PAR !\success \PAR ?\silent)$.
Reduce $\tau(?)$ until it stops before the second $P_2$ of the blocking subsequence with $C_1 = \eempty,C_2 = \full$.
 There are two cases:
 \begin{enumerate}
   \item $\tau(!\success)$ can be reduced until it blocks at a $P_2$.   Then we assume that the process is $\tau(!\success \PAR ?\silent)$       
        Hence  we have a deadlock.
    \item $\tau(!\success)$ can be reduced until it blocks at a $P_1$.    
        This position must be the second position in a blocking subsequence, since reduction starts with $C_1 = \eempty$, 
        and the prefix $\{P_1,T_1\}^*T_2$ enforces that a $T_2$ is executed before any $P_2$ in $\tau(!)$.    
             Due to the form of the blocking sequence   the last step before blocking was a $T_2$.
         We   continue now the reduction  of $\tau(?)$. 
          This can block at a $P_2$, and we will again use a $\tau(!)$-subprocess for unblocking.
          Or it stops at a $P_1$, then we use the $T_1^+$ at the start of a fresh $\tau(!)$ to unblock. 
          Finally, $\tau(?)$ is worked-off. 
            The already used $\tau(!)$ now remain with a prefix $P_1$. 
           We  execute the remaining  $\tau(!)$ until the blocking $P_1$.
  \end{enumerate}
 All cases lead to a deadlock, which is a contradiction to correctness of $\tau$.
\end{proof}
 
 \begin{proposition}\label{prop:p1p1-p2p2-not}
    Blocking type $(P_iP_i{,}P_jP_j)$\,is impossible for a correct translation for $k=2$.
 \end{proposition}
 \begin{proof}   
\Cref{proposition:k-is-2-and-i-not-j} excludes the case $i=j$.
For the case $i \not= j$,
it is sufficient to consider $i=1$, $j=2$ (due to symmetry). Assume that $\tau$ is a correct translation of blocking type
$(P_1P_1,P_2P_2)$.
  \Cref{lemma:blocking-pre-B-all-inital} shows that   $\{T_1,P_1\}^*T_2$  and 
$\{T_1,T_2,P_1,P_2\}^*P_1\{P_2,T_2\}^*T_2P_1$ must be prefixes of $\tau(!)$. Thus $\tau(!)$ must start with $T_1,P_1$ or $T_2$ and the length of $\tau(!)$ is at least 3.
Lemma~\ref{lemma:k2-excl-pref-not-P1} shows that $\tau(!)$ cannot start with $P_1$.
 \Cref{lemma:k-is-2-send-is-not-t1plus-t2,lemma:k-is-2-send-is-not-t1startp2,lemma:k-is-2-send-is-not-t1plus-p1}
show that  
the prefix of $\tau(!)$ cannot be $T_1^+T_2$, $T_1^+P_1$, nor $T_1^*P_2$. Thus $\tau(!)$ cannot start with $T_1$.
\Cref{lemma:k-is-2-send-is-not-t1plus-t2,lemma:k2-T2+P2,lemma:k-is-2-send-is-not-t2plus-p1} show that 
the prefix of $\tau(!)$ cannot be $T_2^+P_2,T_2^+T_1$, nor $T_2^+P_1$.
Thus $\tau(!)$ cannot start with $T_2$.
Hence, we have a contradiction, and $\tau$ cannot be correct.
\end{proof}

 \subsection{\texorpdfstring{Refuting Blocking Types $(P_iP_i,P_i)$,  $(P_i,P_iP_i)$, $(P_i, P_i)$, $(P_iP_i,P_j)$}{Refuting Blocking Types (PiPi,Pi),  (Pi,PiPi), (Pi, Pi), (PiPi,Pj)}}

 \begin{proposition}\label{prop:P1P1P1-not}
    Let $\tau$ be a correct translation. For $k = 2$ the blocking types $(P_1P_1,P_1)$,  $(P_1, P_1P_1)$, and $(P_1, P_1)$ are not possible.
   \end{proposition}
  \begin{proof}
      First, we assume $(P_1P_1,P_1)$. Consider the process $\tau(!\success) \PAR \tau(?\success)$
      which must be must-convergent for a correct
      translation $\tau$.
      The blocking prefix of $\tau(?)$ is of the form $\{P_2,T_2\}^*P_1$, and $IS_1 = \full$. Then construct the following reduction: 
      first, reduce $\tau(!\success)$ until the blocking $P_1$ (now $C_1=\full$ still holds), and then the prefix $\{P_2,T_2\}^*P_1$ of $\tau(?\success)$. 
      If it blocks at some $P_2$, then it is a deadlock,
      and if it blocks at the $P_1$, it is also a deadlock.
      The symmetric type $(P_1, P_1P_1)$ is also impossible (by the symmetric reduction).
      Now assume the type is $(P_1, P_1)$.  Then the blocking prefixes  of $\tau(!)$ and $\tau(?)$ are both of the form $\{P_2,T_2\}^*P_1$.  Reducing $\tau(!)$ blocks at $P_1$.
      Afterwards reducing $\tau(?)$ either stops at a $P_2$, which is a deadlock, or at $P_1$, which is also a deadlock.
      Thus for the must-convergent process $(! \PAR ?\success)$ we
      can construct a reduction sequence for $\tau(!\PAR ?\success)$ that ends in a deadlock. 
  \end{proof}
  
    
In the following,  we only have  to think about  the blocking types $(P_1P_1,P_2)$,  and $(P_1, P_2)$, since $(P_2, P_1P_1)$ is a symmetric case of the first one.

\begin{lemma}\label{lemma:P1P1-P2-not}
 Blocking type $(P_1P_1,P_2)$ is not possible for a correct translation and $k=2$.
 \end{lemma}
 \begin{proof}
 Assume that the blocking type of $\tau$ is $(P_1P_1,P_2)$. 
   \cref{lemma:P1-type-and-store}    
    shows that  $IS_2 = \full$, and the prefix of $\tau(?)$ is $\{P_1,T_1\}^*P_2$.   
    This holds, since if the first symbol in $\tau(?)$ which is in  $\{P_2,T_2\}^*$ is $T_2$, then the blocking type would be different for $\tau(?)$.
  
  Since the blocking type of $\tau(!)$ is $P_1P_1$, \cref{lemma:P1-type-and-store} shows that either $IS_1 = \eempty$
 or the first 1-symbol in the blocking-sequence (which is of the form 
 $R_1P_1\{T_2,P_2\}^*P_1$) is $T_1$.
 %

 The blocking prefix of $\tau(!)$ cannot be  $\{P_1,T_1\}^*$:  This would imply that it stops with $P_1P_1$.
   Then the process $\tau(!\success \PAR ?)$ permits a failing reduction: First, reduce $\tau(?)$ until it blocks with $P_2$, 
  and then reduce $\tau(!)$, which blocks at $P_1$ without changing $C_2$, hence it is a deadlock.  
  \item A prefix of  $\tau(!)$ is of the form  $\{P_1,T_1\}^*T_2$:  Suppose the prefix is $\{P_1,T_1\}^*P_2$.
    Reducing $\tau(!\success \PAR ?)$ as follows: First $\tau(!)$,
   which cannot block within the prefix  $\{P_1,T_1\}^*$, hence it blocks at  $P_2$. Subsequent reduction of $\tau(?)$ leads to a deadlock
   since it blocks at $P_2$.
  
    For the final contradiction, we show that the process $\tau(!\success \PAR ?)$ permits a failing reduction: First, reduce $\tau(?)$ until it blocks with $P_2$, 
      and then reduce $\tau(!)$, which blocks at $P_1$.  If $C_2 = \full$ after the reduction, then it is a deadlock.
      Hence $C_2 = \eempty$ after the reduction. This holds for every reduction of $\tau(!)$ until blocking.
     Now  we restart with the process $\tau(!\success \PAR\ldots \PAR !\success \PAR ?)$, where we  
    will fix the number of $!\success$-subprocesses  later.
     First, reduce $\tau(!)$ until the blocking $P_1$ and get $C_2 = \eempty$. Then we reduce $\tau(?)$ as far as possible.
     There are  cases: 
     \begin{enumerate}
      \item 
 $\tau(?)$ can be completely reduced. Then we reduce the second 
     $\tau(!)$ until a blocking, which will occur at $P_1$.
     Then  $C_1=\full$, and hence both $\tau(!)$ are blocked forever.  
     \item   $\tau(?)$ blocks at a $P_1$, then  we have a deadlock.
       \item  $\tau(?)$ blocks  at a later $P_2$.
         Then again we use the next subprocess $\tau(!)$  and reduce it to the blocking $P_1$, with $C_2 = \eempty$, and can proceed with $\tau(?)$.
         This can be repeated until $\tau(?)$ is completely reduced, where we assume sufficiently many subprocesses $\tau(!\success)$.  
         Finally we get a deadlock by reducing the last  $\tau(!)$ to the blocking, and then we have a deadlock.\qedhere
              \end{enumerate}

  \end{proof}
%
%

 \subsection{\texorpdfstring{Refuting the Blocking Type $(P_1, P_2)$}{Refuting the Blocking Type (P1,P2)}}\label{sec:p1p2}
 The treatment of blocking type $(P_1,P_2)$ requires more arguments.
 We first show a lemma on the suffix of $\tau(!)$  and $\tau(?)$,
 that permit to  reuse results for other initial stores than $(\full,\full)$. 
 
 \begin{lemma}\label{lemma:P1P2-store-simple}
For $k=2$ and a correct translation $\tau$ of blocking type $(P_1, P_2)$,
the initial store can only be $(\full,\full)$ and the prefixes of $\tau(!)$ and $\tau(?)$  are $\{P_2,T_2\}^*P_1$, and  
    $\{P_1,T_1\}^*P_2$.
 \end{lemma}
Due to space constraints the proof of the following proposition is given in \cite{schmidt-schauss-sabel:21-conclock-ext}:
\begin{proposition}\label{prop:building-blocks-no}
 Let $\tau$ be a  translation for $k = 2$ of blocking type $(P_1, P_2)$. 
 Let $\tau(!)$  consist of a sequence of building blocks which follow the pattern
  $\{T_1,T_2\}^*P_1$ or $\{T_1,T_2\}^*P_2$, where in addition a suffix $\{T_1,T_2\}^*$ is appended.
  Let  $\tau(?)$ consist of a sequence of building blocks which follow the pattern
  $\{T_1,T_2\}^*P_1$ or $\{T_1,T_2\}^*P_2$. 
 Then $\tau$ is not correct.
 \end{proposition}


 \begin{corollary}\label{corr:tau-excl-question-T-suffix}
   Let $\tau$ be a correct translation for $k = 2$ of blocking type $(P_1, P_2)$. 
 Then  $\tau(!)$ and $\tau(?)$ have a nontrivial suffix in  $\{T_1,T_2\}^+$.  
 \end{corollary}

Extending a must-convergent process by $! \PAR ?$ may destroy the must-convergence. An example is $!?\silent \PAR ?\success$,
where $! \PAR ? \PAR !?\silent \PAR ?\success$
becomes may-divergent. However, for flat processes, the extension preserves must-convergence, where a \PISIMPLE-process is \emph{flat} if it is of the form $A_1 \PAR \ldots \PAR A_n$,
where $A_i$ is $!\silent, ?\silent, !\success$, 
or $?\success$. 
\begin{lemma}\label{lemma:must-conv-kept-flat}
     Let $Q$ be a flat~\PISIMPLE-process that is must-convergent. Then the process $! \PAR ? \PAR  Q$ is also must-convergent.
\end{lemma}

\begin{proposition}
\label{prop:no-blockin-type-p1-p2}
Blocking type $(P_1, P_2)$ is impossible for correct translations for $k = 2$.
 \end{proposition}
\ignore{
\begin{proposition}\label{prop:P1-P2-not}
  Let $\tau$ be a  translation for $k = 2$ of blocking type $(P_1, P_2)$. 
  Then the translation $\tau$ is not correct.
 \end{proposition}
 }
\begin{proof} 
Assume  that $\tau$ is correct for initial state (\full,\full).
 Then \cref{corr:tau-excl-question-T-suffix} shows that
 $\tau(!)$ and $\tau(?)$ must end with $\{T_1,T_2\}^+$. Since $\tau(! \PAR ?)$ must be completely
executable (see \cref{lemma:reduction-uses-all-symbols}),  reducing
$\tau((!\PAR ? \PAR Q),(\full,\full)) \xrightarrow{LS,*}
(\tau(Q),(k_1,k_2))$ must lead to a state $(k_1,k_2) \not= (\full,\full)$ for every $Q$.
We consider the blocking behavior of $\tau$ for $(k_1,k_2) \not= (\full,\full)$.
\begin{itemize}
 \item If $\tau(?)$ is non-blocking for $(k_1,k_2)$, then consider the must-divergent process $!?\success \PAR ?$.
 Then $(\tau(!?\success \PAR ?),(\full,\full))
 \xrightarrow{LS,*} (\tau(?)\success,(k_1,k_2))
 \xrightarrow{LS,*} (\success,(l_1,l_2))$.
 Thus $\tau$ is not  correct.
 \item If $\tau(!)$ is non-blocking for $(k_1,k_2)$, then consider the must-divergent process $?!\success \PAR !$.
 Then $(\tau(?!\success \PAR !),(\full,\full))
 \xrightarrow{LS,*} (\tau(!)\success,(k_1,k_2))
 \xrightarrow{LS,*} (\success,(l_1,l_2))$.
 Thus $\tau$ is not  correct.
\item  We know that the prefix of $\tau(!)$ cannot be $T_1^+T_2$ nor $T_2^+T_1$ (see \cref{lemma:k-is-2-send-is-not-t1plus-t2}). 
 \item The blocking type of $\tau$ for $(k_1,k_2)$
  is $(P_iP_i,P_jP_j)$.
  Then the proof of \cref{proposition:k-is-2-and-i-not-j} can be adapted to first show that $i \not=j$: It uses flat must-convergent processes  and constructs failing reductions. 
  Let $Q$ be such a counter-example process
  \cref{lemma:must-conv-kept-flat} shows 
  that $! \PAR ? \PAR Q$ is also must-convergent,
 and thus $\tau(! \PAR ? \PAR Q, (\full,\full)) \xrightarrow{LS,*} (\tau(Q),(k_1,k_2))$ and thus
 $(\tau(Q),(k_1,k_2))$ also must be must-convergent. But the constructed failing reductions of \cref{proposition:k-is-2-and-i-not-j} refute this.
  For the case $i \not= j$, we can reason as in the 
  lemmas before 
     \cref{prop:p1p1-p2p2-not} and also as in
     \cref{prop:p1p1-p2p2-not} itself, since they all  use  flat must-convergent \PISIMPLE-processes and show that there are failing reductions after translating them.
     Again if $Q$ is such a process, 
  \cref{lemma:must-conv-kept-flat} shows 
  that $! \PAR ? \PAR Q$ is also must-convergent,
 and thus $\tau(! \PAR ? \PAR Q, (\full,\full)) \xrightarrow{LS,*} (\tau(Q),(k_1,k_2))$. Thus $(\tau(Q),(k_1,k_2))$  must be must-convergent. But the constructed failing reductions
 in the proofs in the lemmas before 
 \cref{prop:p1p1-p2p2-not}, or in the proof of 
 \cref{prop:p1p1-p2p2-not}, respectively, refute the must-convergence. Thus the proved properties also hold if $\tau$ is of blocking type 
  $(P_iP_i,P_jP_j)$ for $(k_1,k_2)$ (where \cref{lemma:k-is-2-send-is-not-t1plus-t2} can be used directly, since it holds for any initial state).
       This shows $(P_iP_i, P_jP_j)$ is impossible as blocking type of $\tau$ for $(k_1,k_2)$. 
 \item  The blocking type of $\tau$ for $(k_1,k_2)$ is $(P_1P_1,P_1)$ or $(P_1,P_1P_1)$ or $(P_1,P_1)$. Then the must-convergent $\PISIMPLE$-processes in the proof of  \Cref{prop:P1P1P1-not}
     can be used, since they are flat.
     Let $Q$ be such a process. By 
       \cref{lemma:must-conv-kept-flat}
        $! \PAR ? \PAR Q$ is must-convergent.
        Since $\tau$ is correct $\tau(! \PAR ? \PAR Q)$ is must-convergent and thus
        $(\tau(Q),(k_1,k_2))$ is must-convergent.
     The proof of \cref{prop:P1P1P1-not} shows that $(\tau(Q),(k_1,k_2))$ may-diverges,  a contradiction.
    %
  \item $\tau$ is of blocking type $(P_1P_1,P_2)$ for $(k_1,k_2)$.
            Then the reasoning is analogous to the previous case using the must-convergent flat counterexample processes of 
            \cref{lemma:P1P1-P2-not}.
 \item The blocking type $(P_1,P_2)$ is not possible,
 since we have a store $(k_1,k_2) \not= (\full,\full)$. \qedhere
 
 

\end{itemize}
\end{proof}

We now prove the main result:
 \begin{theorem}\label{thm:tau-2-incorrect}
   Let $IS$ be an initial store with two elements, and $\tau:$ $\PISIMPLE~\to \LOCKSIMPLE_{2,IS}$ be a compositional translation.
   Then $\tau$ is not correct.
  \end{theorem}
\begin{proof}
The proof is structured along the blocking types (\cref{def:blocking-types}) of translations. For $k = 2$ there are 4 blocking types 
of subprocesses,  and
16 potentially possible  blocking types of translations.  
Proposition \ref{proposition:k-is-2-and-i-not-j}  shows that type $(P_iP_i,P_iP_i)$ is impossible, and Proposition
\ref{prop:p1p1-p2p2-not} that $(P_iP_i,P_jP_j)$ for $i \not= j$ is impossible.
Proposition \ref{prop:P1P1P1-not} shows that blocking types $(P_1P_1,P_1)$,  $(P_1, P_1P_1)$, and $(P_1, P_1)$ are impossible,
and also the same for $P_2$, since this is analogous.
Lemma \ref{lemma:P1P1-P2-not}  shows that blocking types $(P_1P_1,P_2)$ (and also $(P_2P_2,P_1)$, $(P_1,P_2P_2), (P_2.P_1P_1)$ are impossible.
The harder case $(P_1,P_2)$ (and the symmetric case $(P_2,P_1)$)  is shown in a series of lemmas and finally 
proved in Proposition \ref{prop:no-blockin-type-p1-p2}.      
\end{proof}

  \section{Conclusion}\label{section:conclusion}
    We proved that for locks where exactly one of the operations (put or take) blocks if the store is not as expected,
    a correct translation from $\PISIMPLE$ into $\LOCKSIMPLE$ requires at least three locks, and also exhibited a correct translation for 
    three locks. 
   It remains open whether for all the considered blocking variants 
   and initial storage values there are correct translations for $k \geq 3$.
%
%
     Future work is to provide more arguments that our results  
   can be transferred to 
     full concurrent programming languages.
   Future work is also  to investigate the same questions for locks
    where both, put and take are blocking, if
   the store is not as expected (like MVars in Concurrent Haskell).
%
%
 
 


\bibliographystyle{eptcs}
\bibliography{pichfbiburl.bib}

\begin{thebibliography}{10}
\providecommand{\bibitemdeclare}[2]{}
\providecommand{\surnamestart}{}
\providecommand{\surnameend}{}
\providecommand{\urlprefix}{Available at }
\providecommand{\url}[1]{\texttt{#1}}
\providecommand{\href}[2]{\texttt{#2}}
\providecommand{\urlalt}[2]{\href{#1}{#2}}
\providecommand{\doi}[1]{doi:\urlalt{http://dx.doi.org/#1}{#1}}
\providecommand{\bibinfo}[2]{#2}

\bibitemdeclare{techreport}{boudol:1992}
\bibitem{boudol:1992}
\bibinfo{author}{G\'erard \surnamestart Boudol\surnameend}
  (\bibinfo{year}{1992}): \emph{\bibinfo{title}{Asynchrony and the
  Pi-calculus}}.
\newblock \bibinfo{type}{Technical Report} \bibinfo{number}{Research Report
  RR-1702,inria-00076939}, \bibinfo{institution}{INRIA, France}.
\newblock \urlprefix\url{https://hal.inria.fr/inria-00076939}.

\bibitemdeclare{inproceedings}{Chaudhuri09}
\bibitem{Chaudhuri09}
\bibinfo{author}{Avik \surnamestart Chaudhuri\surnameend}
  (\bibinfo{year}{2009}): \emph{\bibinfo{title}{A concurrent {ML} library in
  concurrent Haskell}}.
\newblock In: {\sl \bibinfo{booktitle}{{ICFP} 2009}},
  \bibinfo{publisher}{{ACM}}, pp. \bibinfo{pages}{269--280},
  \doi{10.1145/1596550.1596589}.

\bibitemdeclare{inproceedings}{GlabbeekGLM19}
\bibitem{GlabbeekGLM19}
\bibinfo{author}{Rob \surnamestart van Glabbeek\surnameend},
  \bibinfo{author}{Ursula \surnamestart Goltz\surnameend},
  \bibinfo{author}{Christopher \surnamestart Lippert\surnameend} \&
  \bibinfo{author}{Stephan \surnamestart Mennicke\surnameend}
  (\bibinfo{year}{2019}): \emph{\bibinfo{title}{Stronger Validity Criteria for
  Encoding Synchrony}}.
\newblock In: {\sl \bibinfo{booktitle}{The Art of Modelling Computational
  Systems: {A} Journey from Logic and Concurrency to Security and Privacy -
  Essays Dedicated to Catuscia Palamidessi on the Occasion of Her 60th
  Birthday}}, {\sl \bibinfo{series}{LNCS}} \bibinfo{volume}{11760},
  \bibinfo{publisher}{Springer}, pp. \bibinfo{pages}{182--205},
  \doi{10.1007/978-3-030-31175-9\_11}.

\bibitemdeclare{article}{Gorla:10}
\bibitem{Gorla:10}
\bibinfo{author}{Daniele \surnamestart Gorla\surnameend}
  (\bibinfo{year}{2010}): \emph{\bibinfo{title}{Towards a unified approach to
  encodability and separation results for process calculi}}.
\newblock {\sl \bibinfo{journal}{Inf. Comput.}}
  \bibinfo{volume}{208}(\bibinfo{number}{9}), pp. \bibinfo{pages}{1031--1053},
  \doi{10.1016/j.ic.2010.05.002}.

\bibitemdeclare{inproceedings}{Honda:1991}
\bibitem{Honda:1991}
\bibinfo{author}{Kohei \surnamestart Honda\surnameend} \&
  \bibinfo{author}{Mario \surnamestart Tokoro\surnameend}
  (\bibinfo{year}{1991}): \emph{\bibinfo{title}{An Object Calculus for
  Asynchronous Communication}}.
\newblock In: {\sl \bibinfo{booktitle}{Proceedings of the European Conference
  on Object-Oriented Programming}}, \bibinfo{series}{ECOOP '91},
  \bibinfo{publisher}{Springer-Verlag}, pp. \bibinfo{pages}{133--147},
  \doi{10.1007/BFb0057019}.

\bibitemdeclare{article}{milner1992calculus}
\bibitem{milner1992calculus}
\bibinfo{author}{Robin \surnamestart Milner\surnameend},
  \bibinfo{author}{Joachim \surnamestart Parrow\surnameend} \&
  \bibinfo{author}{David \surnamestart Walker\surnameend}
  (\bibinfo{year}{1992}): \emph{\bibinfo{title}{A calculus of mobile processes,
  {I}}}.
\newblock {\sl \bibinfo{journal}{Information and computation}}
  \bibinfo{volume}{100}(\bibinfo{number}{1}), pp. \bibinfo{pages}{1--40},
  \doi{10.1016/0890-5401(92)90008-4}.

\bibitemdeclare{article}{jlambda-fut}
\bibitem{jlambda-fut}
\bibinfo{author}{Joachim \surnamestart Niehren\surnameend},
  \bibinfo{author}{Jan \surnamestart Schwinghammer\surnameend} \&
  \bibinfo{author}{Gert \surnamestart Smolka\surnameend}
  (\bibinfo{year}{2006}): \emph{\bibinfo{title}{A Concurrent Lambda Calculus
  with Futures}}.
\newblock {\sl \bibinfo{journal}{{Theoretical Computer Science}}}
  \bibinfo{volume}{364}(\bibinfo{number}{3}), pp. \bibinfo{pages}{{338--356}},
  \doi{10.1016/j.tcs.2006.08.016}.

\bibitemdeclare{inproceedings}{Palamidessi:97}
\bibitem{Palamidessi:97}
\bibinfo{author}{Catuscia \surnamestart Palamidessi\surnameend}
  (\bibinfo{year}{1997}): \emph{\bibinfo{title}{Comparing the Expressive Power
  of the Synchronous and the Asynchronous pi-calculus}}.
\newblock In: {\sl \bibinfo{booktitle}{{POPL}~1997}}, \bibinfo{publisher}{{ACM}
  Press}, pp. \bibinfo{pages}{256--265}, \doi{10.1145/263699.263731}.

\bibitemdeclare{article}{Palamidessi03}
\bibitem{Palamidessi03}
\bibinfo{author}{Catuscia \surnamestart Palamidessi\surnameend}
  (\bibinfo{year}{2003}): \emph{\bibinfo{title}{Comparing The Expressive Power
  Of The Synchronous And Asynchronous Pi-Calculi}}.
\newblock {\sl \bibinfo{journal}{Math. Structures Comput. Sci.}}
  \bibinfo{volume}{13}(\bibinfo{number}{5}), pp. \bibinfo{pages}{685--719},
  \doi{10.1017/S0960129503004043}.

\bibitemdeclare{inproceedings}{peyton-gordon-finne:96}
\bibitem{peyton-gordon-finne:96}
\bibinfo{author}{Simon~L. \surnamestart {Peyton Jones}\surnameend},
  \bibinfo{author}{Andrew \surnamestart Gordon\surnameend} \&
  \bibinfo{author}{Sigbjorn \surnamestart Finne\surnameend}
  (\bibinfo{year}{1996}): \emph{\bibinfo{title}{Concurrent {H}askell}}.
\newblock In: {\sl \bibinfo{booktitle}{{POPL 1996}}}, \bibinfo{publisher}{ACM},
  pp. \bibinfo{pages}{295--308}, \doi{10.1145/237721.237794}.

\bibitemdeclare{article}{rensink-vogler:07}
\bibitem{rensink-vogler:07}
\bibinfo{author}{Arend \surnamestart Rensink\surnameend} \&
  \bibinfo{author}{Walter \surnamestart Vogler\surnameend}
  (\bibinfo{year}{2007}): \emph{\bibinfo{title}{Fair testing}}.
\newblock {\sl \bibinfo{journal}{Inform. and Comput.}}
  \bibinfo{volume}{205}(\bibinfo{number}{2}), pp. \bibinfo{pages}{125--198},
  \doi{10.1016/j.ic.2006.06.002}.

\bibitemdeclare{inproceedings}{Russell01}
\bibitem{Russell01}
\bibinfo{author}{George \surnamestart Russell\surnameend}
  (\bibinfo{year}{2001}): \emph{\bibinfo{title}{Events in {H}askell, and How to
  Implement Them}}.
\newblock In: {\sl \bibinfo{booktitle}{{ICFP}~2001}},
  \bibinfo{publisher}{{ACM}}, pp. \bibinfo{pages}{157--168},
  \doi{10.1145/507635.507655}.

\bibitemdeclare{article}{sabel-schmidt-schauss-MSCS:08}
\bibitem{sabel-schmidt-schauss-MSCS:08}
\bibinfo{author}{David \surnamestart Sabel\surnameend} \&
  \bibinfo{author}{Manfred \surnamestart Schmidt-Schau{\ss}\surnameend}
  (\bibinfo{year}{2008}): \emph{\bibinfo{title}{A Call-by-Need Lambda-Calculus
  with Locally Bottom-Avoiding Choice: Context Lemma and Correctness of
  Transformations}}.
\newblock {\sl \bibinfo{journal}{Math. Structures Comput. Sci.}}
  \bibinfo{volume}{18}(\bibinfo{number}{03}), pp. \bibinfo{pages}{501--553},
  \doi{10.1017/S0960129508006774}.

\bibitemdeclare{inproceedings}{sabel-schmidt-schauss-PPDP:2011}
\bibitem{sabel-schmidt-schauss-PPDP:2011}
\bibinfo{author}{David \surnamestart Sabel\surnameend} \&
  \bibinfo{author}{Manfred \surnamestart Schmidt-Schau{\ss}\surnameend}
  (\bibinfo{year}{2011}): \emph{\bibinfo{title}{A contextual semantics for
  {C}oncurrent {H}askell with futures}}.
\newblock In: {\sl \bibinfo{booktitle}{{PPDP 2011}}}, \bibinfo{publisher}{ACM},
  pp. \bibinfo{pages}{101--112}, \doi{10.1145/2003476.2003492}.

\bibitemdeclare{inproceedings}{sabel-schmidt-schauss-LICS:12}
\bibitem{sabel-schmidt-schauss-LICS:12}
\bibinfo{author}{David \surnamestart Sabel\surnameend} \&
  \bibinfo{author}{Manfred \surnamestart Schmidt-Schau{\ss}\surnameend}
  (\bibinfo{year}{2012}): \emph{\bibinfo{title}{Conservative Concurrency in
  {H}askell}}.
\newblock In: {\sl \bibinfo{booktitle}{{LICS 2012}}},
  \bibinfo{publisher}{IEEE}, pp. \bibinfo{pages}{561--570},
  \doi{10.1109/LICS.2012.66}.

\bibitemdeclare{book}{sangiorgi-walker:01}
\bibitem{sangiorgi-walker:01}
\bibinfo{author}{Davide \surnamestart Sangiorgi\surnameend} \&
  \bibinfo{author}{David \surnamestart Walker\surnameend}
  (\bibinfo{year}{2001}): \emph{\bibinfo{title}{The $\pi$-calculus: a theory of
  mobile processes}}.
\newblock \bibinfo{publisher}{Cambridge university press}.

\bibitemdeclare{inproceedings}{schmidt-schauss-niehren-schwinghammer-sabel-ifip-tcs:08}
\bibitem{schmidt-schauss-niehren-schwinghammer-sabel-ifip-tcs:08}
\bibinfo{author}{Manfred \surnamestart Schmidt-Schau{\ss}\surnameend},
  \bibinfo{author}{Joachim \surnamestart Niehren\surnameend},
  \bibinfo{author}{Jan \surnamestart Schwinghammer\surnameend} \&
  \bibinfo{author}{David \surnamestart Sabel\surnameend}
  (\bibinfo{year}{2008}): \emph{\bibinfo{title}{Adequacy of Compositional
  Translations for Observational Semantics}}.
\newblock In: {\sl \bibinfo{booktitle}{{IFIP TCS 2008}}}, {\sl
  \bibinfo{series}{IFIP}} \bibinfo{volume}{273}, \bibinfo{publisher}{Springer},
  pp. \bibinfo{pages}{521--535}, \doi{10.1007/978-0-387-09680-3\_35}.

\bibitemdeclare{article}{schmidt-schauss-sabel-maymust:10}
\bibitem{schmidt-schauss-sabel-maymust:10}
\bibinfo{author}{Manfred \surnamestart Schmidt-Schau{\ss}\surnameend} \&
  \bibinfo{author}{David \surnamestart Sabel\surnameend}
  (\bibinfo{year}{2010}): \emph{\bibinfo{title}{Closures of may-, should- and
  must-convergences for contextual equivalence}}.
\newblock {\sl \bibinfo{journal}{Inform. and Comput.}}
  \bibinfo{volume}{110}(\bibinfo{number}{6}), pp. \bibinfo{pages}{232 -- 235},
  \doi{10.1016/j.ipl.2010.01.001}.

\bibitemdeclare{inproceedings}{schmidt-schauss-sabel-2020}
\bibitem{schmidt-schauss-sabel-2020}
\bibinfo{author}{Manfred \surnamestart Schmidt-Schau{\ss}\surnameend} \&
  \bibinfo{author}{David \surnamestart Sabel\surnameend}
  (\bibinfo{year}{2020}): \emph{\bibinfo{title}{Correctly Implementing
  Synchronous Message Passing in the Pi-Calculus By Concurrent Haskell's
  MVars}}.
\newblock In: {\sl \bibinfo{booktitle}{{EXPRESS/SOS 2020}}}, {\sl
  \bibinfo{series}{Electronic Proceedings in Theoretical Computer Science}}
  \bibinfo{volume}{322}, \bibinfo{publisher}{Open Publishing Association}, pp.
  \bibinfo{pages}{88--105}, \doi{10.4204/EPTCS.322.8}.

\bibitemdeclare{misc}{schmidt-schauss-sabel-wpte:2020}
\bibitem{schmidt-schauss-sabel-wpte:2020}
\bibinfo{author}{Manfred \surnamestart Schmidt-Schau{\ss}\surnameend} \&
  \bibinfo{author}{David \surnamestart Sabel\surnameend}
  (\bibinfo{year}{2020}): \emph{\bibinfo{title}{On Impossibility of Simple
  Translations of Concurrent Calculi}}.
\newblock \bibinfo{note}{Presented at {WPTE 2020}, pre-proceedings available
  via \href{http://maude.ucm.es/wpte20/}{http://maude.ucm.es/wpte20/}}.

\bibitemdeclare{inproceedings}{schauss-sabel-dallmeyer:18}
\bibitem{schauss-sabel-dallmeyer:18}
\bibinfo{author}{Manfred \surnamestart Schmidt{-}Schau{\ss}\surnameend},
  \bibinfo{author}{David \surnamestart Sabel\surnameend} \&
  \bibinfo{author}{Nils \surnamestart Dallmeyer\surnameend}
  (\bibinfo{year}{2018}): \emph{\bibinfo{title}{Sequential and Parallel
  Improvements in a Concurrent Functional Programming Language}}.
\newblock In: {\sl \bibinfo{booktitle}{{PPDP 2018}}},
  \bibinfo{publisher}{{ACM}}, pp. \bibinfo{pages}{20:1--20:13},
  \doi{10.1145/3236950.3236952}.

\bibitemdeclare{article}{schmidtschauss-sabel-niehren-schwing-tcs:15}
\bibitem{schmidtschauss-sabel-niehren-schwing-tcs:15}
\bibinfo{author}{Manfred \surnamestart Schmidt{-}Schau{\ss}\surnameend},
  \bibinfo{author}{David \surnamestart Sabel\surnameend},
  \bibinfo{author}{Joachim \surnamestart Niehren\surnameend} \&
  \bibinfo{author}{Jan \surnamestart Schwinghammer\surnameend}
  (\bibinfo{year}{2015}): \emph{\bibinfo{title}{Observational program calculi
  and the correctness of translations}}.
\newblock {\sl \bibinfo{journal}{Theor. Comput. Sci.}} \bibinfo{volume}{577},
  pp. \bibinfo{pages}{98--124}, \doi{10.1016/j.tcs.2015.02.027}.

\bibitemdeclare{article}{schmidt-schauss-sabel:21-conclock-ext}
\bibitem{schmidt-schauss-sabel:21-conclock-ext}
\bibinfo{author}{Manfred \surnamestart Schmidt-Schauß\surnameend} \&
  \bibinfo{author}{David \surnamestart Sabel\surnameend}
  (\bibinfo{year}{2021}): \emph{\bibinfo{title}{Minimal Translations from
  Synchronous Communication to Synchronizing Locks (Extended Version)}}.
\newblock {\sl \bibinfo{journal}{CoRR}} \bibinfo{volume}{abs/2107.14651}.
\newblock \urlprefix\url{https://arxiv.org/abs/2107.14651}.

\bibitemdeclare{inproceedings}{schwinghammer-sabel-schmidt-schauss-niehren:09:ml}
\bibitem{schwinghammer-sabel-schmidt-schauss-niehren:09:ml}
\bibinfo{author}{Jan \surnamestart Schwinghammer\surnameend},
  \bibinfo{author}{David \surnamestart Sabel\surnameend},
  \bibinfo{author}{Manfred \surnamestart Schmidt-Schau{\ss}\surnameend} \&
  \bibinfo{author}{Joachim \surnamestart Niehren\surnameend}
  (\bibinfo{year}{2009}): \emph{\bibinfo{title}{Correctly translating
  concurrency primitives}}.
\newblock In: {\sl \bibinfo{booktitle}{{ML 2009}}}, \bibinfo{publisher}{ACM},
  pp. \bibinfo{pages}{27--38}, \doi{10.1145/1596627.1596633}.

\end{thebibliography}
 
\end{document}